\documentclass[11pt]{article}
\usepackage[twoside,a4paper,margin=1in]{geometry}
\usepackage[latin1]{inputenc}
\usepackage{amsmath,amssymb}

\usepackage{amsthm}
\usepackage{hyperref,algorithmic}
\usepackage{graphicx,url,amsmath,amsfonts,amssymb,subfigure,bbm,bm}
\usepackage{color}
\usepackage{verbatim}
\usepackage{wrapfig}

\newtheorem{theorem}{Theorem}[section]
\newtheorem{lemma}[theorem]{Lemma}

\newtheorem{corollary}[theorem]{Corollary}

\newtheorem{proposition}[theorem]{Proposition}

\newcommand{\executeiffilenewer}[3]{%
\ifnum\pdfstrcmp{\pdffilemoddate{#1}}%
{\pdffilemoddate{#2}}>0%
{\immediate\write18{#3}}\fi%
}

\graphicspath{{Fig/}}

\definecolor{blueblack}{rgb}{0,0,.7}
\newcommand{\emphdef}[1]{%
  \textcolor{blueblack}{%
    \textbf{\emph{#1}}%
  }%
}

\newcommand{\OPT}{\textrm{OPT}}
\newcommand{\dist}{\textrm{dist}}

 \makeatletter
 \def\cramped
   {\parskip\@outerparskip\@topsep\parskip
   \@topsepadd2pt\itemsep0pt
 }
 \makeatother

\AtBeginDocument{
\abovedisplayskip 2pt plus 2pt
\belowdisplayskip 2pt plus 2pt
}

\title{A Fixed Parameter Tractable Approximation Scheme for the
  Optimal Cut Graph of a Surface\thanks{The research of the second
    author leading to these results has received funding from the People Programme
(Marie Curie Actions) of the European Union's Seventh Framework Programme
(FP7/2007-2013) under REA grant agreement n° [291734].}}%
\author{Vincent Cohen-Addad\thanks{D\'epartement d'informatique,
    \'Ecole normale sup\'erieure, Paris, France. Email: \protect\url{vcohen@di.ens.fr}}
\and
Arnaud de Mesmay\thanks{IST Austria, Klosterneuburg, Austria. Email: \protect\url{arnaud.de.mesmay@ist.ac.at}}}

\date{}

\begin{document}

\maketitle

\begin{abstract}
Given a graph $G$ cellularly embedded on a surface $\Sigma$ of genus $g$, a cut graph is a
subgraph of $G$ such that cutting $\Sigma$ along $G$ yields a topological
disk. We provide a fixed parameter tractable approximation scheme for
the problem of computing the shortest cut graph, that is, for any
$\varepsilon >0$, we show how to compute a $(1+ \varepsilon)$ approximation
of the shortest cut graph in time $f(\varepsilon, g)n^3$. 

Our techniques first rely on the computation of a spanner for the
problem using the technique of brick decompositions, to reduce the problem to the
case of bounded tree-width. Then, to solve the bounded tree-width case, we
introduce a variant of the surface-cut decomposition of Ru\'e, Sau and
Thilikos, which may be of independent interest.
\end{abstract}

\section{Introduction}

Embedded graphs are commonly used to model a wide array of discrete
structures, and in many cases it is necessary to consider embeddings
into surfaces instead of the plane or the sphere. For example, many
instances of network design actually feature some crossings, coming
from tunnels or overpasses, which are appropriately modeled by a
surface of small genus. In other settings, such as in computer
graphics or computer-aided design, we are looking for a discrete model
for objects which inherently display a non-trivial topology (e.g.,
holes), and graphs embedded on surfaces are the natural tool for
that. From a more theoretical point of view, the graph structure theorem of
Robertson and Seymour showcases a very strong connection between
graphs embedded on surfaces and minor-closed families of graphs.

When dealing with embedded graphs, a classical problem, to which a lot
of effort has been devoted in the past decade, is to find a
\textit{topological decomposition} of the underlying surface, i.e., to
cut the surface into simpler pieces so as to simplify its topology, or
equivalently to cut the embedded graph into a planar graph, see the
recent surveys~\cite{c-tags-12,e-cocb-12}. This is a fundamental
operation in algorithm design for surface-embedded graphs, as it
allows to apply the vast number of tools available for planar graphs
to this more general setting. Furthermore, making a graph planar is
useful for various purposes in computer graphics and mesh processing, see
for example~\cite{whds-retfi-04}. No matter the
application, a crucial parameter is always the length of the
topological decomposition: having good control on it ensures that the
meaningful features of the embedded graphs did not get too much
distorted during the cutting.

In this article, we are interested in the problem of computing a short
cut graph: For a graph $G$ with $n$ vertices embedded on a surface $\Sigma$ of genus $g$, a \emphdef{cut graph} of $G$ is a subgraph $C \subseteq
G$ such that cutting $\Sigma$ along $C$ gives a topological disk. The
problem of computing the shortest possible cut graph of an embedded
graph was introduced by Erickson and Har-Peled~\cite{eh-ocsd-04}, who
showed that it is NP-hard, provided an $n^{O(g)}$ algorithm to
compute it, as well as an $O(g^2 n \log n)$ algorithm to compute a $O(\log^2 g)$ approximation. Now, since in most
practical applications, the genus of the embedded graph tends to be
quite small compared to the complexity of the graph, it is natural to also
investigate this problem through the lens of parametrized
  complexity, which provides a natural framework to study the
dependency of cutting algorithms with respect to the genus. In this
direction, Erickson and Har-Peled asked whether computing the shortest
cut graph is \textit{fixed-parameter tractable}, i.e. whether it can
be solved in time $f(g) n^{O(1)}$ for some function $f$. This
question is, up to our knowledge, still open, and we address here the
neighborly problem of devising a good approximation algorithm working
in fixed parameter tractable time with respect to the genus; we refer to the survey of Marx~\cite{m-pcaa-08}
for more background on these algorithms at the
intersection of approximation algorithms and parametrized complexity.

\paragraph{Our results.} In this article, we provide a
\textit{fixed-parameter tractable approximation scheme} for the
problem of computing the shortest cut graph of an embedded graph. Namely, we prove the following theorem.

\begin{theorem}\label{T:main} Let $G$ be a weighted graph cellularly embedded on a surface
  $\Sigma$ of genus $g$. For any $\varepsilon >0$,
  there exists an algorithm computing a
  $(1+\varepsilon)$-approximation of the shortest cut graph of $G$,
  which runs in time $f(\varepsilon,g)n^{3}$ for some function $f$.
\end{theorem}

\paragraph{Our techniques.} Our algorithm uses the brick
decompositions of Borradaile, Klein and Mathieu~\cite{bkm-nasst-09} for
subset-connectivity problems in planar graphs, which have been extended
to bounded genus graphs by Borradaile, Demaine and
Tazari \cite{bdt-ptass-13}. Although brick decompositions are now a
common tool for optimization problems for embedded graphs, it is to
our knowledge the first time they are applied to compute
topological decompositions. In a
nutshell, the idea is the following:

\begin{enumerate}\cramped
\item We first compute a \textit{spanner} $G_{span}$ for our problem, namely a
  subgraph of the input graph containing a
  $(1+\varepsilon)$-approximation of the optimal cut graph, and having
  total length bounded by $f(g,\varepsilon)$ times the length of the
  optimal cut graph, for some function $f$. This is achieved via \textit{brick decompositions}.

\item Using a result of \textit{contraction-decomposition} of Demaine, Hajiaghayi and Mohar~\cite{dhm-aacd-10}, we
  \textit{contract} a set of edges of controlled length in $G_{span}$,
  obtaining a graph $G_{tw}$ of bounded tree-width.
\item We use dynamic programming on $G_{tw}$ to compute its
  optimal cut graph.
\item We incorporate back the contracted edges, which gives us a
  subgraph of $G$ cutting the surface into one or more disks. Removing
  edges so that the complement is a single disk gives our final cut graph.
\end{enumerate}

The first steps of this framework mostly follow from the same
techniques as in the article of
Borradaile et al.~\cite{bdt-ptass-13}, the only difference being that we need
a specific structure theorem to show that the obtained graph is indeed
a spanner for our problem. However, as the restriction of a cut graph
to a brick, i.e., a small disk on the surface, is a forest, this
structure theorem is a variation of an existing theorem
for the Steiner tree problem~\cite{bkm-nasst-09}.

The main difficulty of this approach lies instead in the third
step. Since a cut graph is inherently a topological notion, it is key
for a dynamic programming approach to work with a tree-decomposition
having nice topological properties. An appealing concept has been
developed by Ru\'e, Sau and Thilikos~\cite{rst-dpgs-14} for the neighborly (and for our
purpose, equivalent) notion of branch-decomposition: they introduced
\textit{surface-cut decompositions} with this exact goal of giving a
nice topological structure to work with when designing dynamic
programs for graph on surfaces (see also Bonsma~\cite{b-ssdsi-12} for 
a related concept). However, their approach is cumbersome
for our purpose when the graph embeddings are not \textit{polyhedral}
(we refer to the introduction for precise definitions), as it first
relies on computing a \textit{polyhedral decomposition} of the input
graph.
While dynamic programming over these polyhedral decompositions can be
achieved for the class of problems that they consider, it seems
unclear how to do it for the problem of computing a shortest cut
graph.

We propose two ways to circumvent this issue. In the first one, we
observe that the need for polyhedral embeddings in surface-cut
decompositions can be traced back exclusively to a theorem of Fomin
and
Thilikos~\cite[Lemma~5.1]{rst-dpgs-14}\cite[Theorem~1]{ft-sdppge-07}
relating the branch-width of an embedded graph and the carving-width
of its medial graph, the proof of which uses crucially that the graph
embedding is polyhedral. But another proof of this theorem which does 
not rely on this assumption was obtained by
Inkmann~\cite[Theorem~3.6.1]{i-tpdgs-07}. 
Therefore, the full strength of surface cut decompositions
can be used without first relying on polyhedral decompositions.

However, since Inkmann's proof is intricate and  has never been published 
we also propose an alternative, self-contained, solution
tailored to our problem. For our purpose, it is enough to make the
graph polyhedral at the end of the second step of the framework while
preserving a strong bound on the branch-width of the graph, we
show that this can be achieved by superposing medial graphs and
triangulating with care. With appropriate heavy weights on the new
edges, we can ensure that they do not impact the length of the optimal
cut graph and that we still obtain a valid solution to our problem.

Finally, both approaches allow us to work with a branch decomposition
that possesses a nice topological structure. 
We then show how to exploit it to write a dynamic program
to compute the shortest cut graph in fixed parameter tractable time
for graphs of bounded tree-width.

\paragraph{Organization of the paper.} We start by introducing the
main notions surrounding embedded graphs and brick decompositions in
Section~\ref{S:prelim}. We then prove the structure theorem in
Section~\ref{S:struc}, showing that the brick decomposition with
portals contains a cut graph which is at most (1+$\varepsilon$) longer
than the optimal one. In Section~\ref{S:algorithm}, we show how to
combine this structure theorem with the aforementioned framework to
obtain our algorithm. This algorithm relies on one that solves the
problem when the input graph has bounded tree-width, which is
described in Section~\ref{S:tw}.

\section{Preliminaries}\label{S:prelim}

All graphs $G=(V,E)$ in this article
are multigraphs, possibly with loops, have $n$ vertices, $m$ edges, are undirected and their edges are
weighted with a length $\ell(e)$. These weights induce naturally a length
on paths and subgraphs of $G$. 

\paragraph{Graphs on surfaces.} We will be using classical notions of graphs embedded on surfaces, for
more background on the subject, we refer to the textbook of Mohar and
Thomassen~\cite{mt-gs-01}. Throughout the article, $\Sigma$ will denote a
compact connected surface of Euler genus $g$, which we will simply
call genus. An \emphdef{embedding} of $G$ on $\Sigma$ is a crossing-free
drawing of $G$ on $\Sigma$, i.e. the images of the vertices are pairwise
distinct and the image of each edge is a simple path intersecting the
image of no other vertex or edge, except possibly at its endpoints. We
will always identify an abstract graph with its embedding. A
\emphdef{face} of the embedding is a connected component of the
complement of the graph. A \emphdef{cellular embedding} is an
embedding of a graph where every face is a topological disk. Every
embedding in this paper will be assumed to be cellular. A graph
embedding is a triangulation if all the faces have degree
three. \emphdef{Euler's formula} states that for a graph $G$ embedded
on a surface $\Sigma$, we have $n-m+f=2-g$, for $f$ the number of faces of
the embedding. A \emphdef{noose} is an embedding of the circle
$\mathbb{S}^1$ on $\Sigma$ which intersects $G$ only at its vertices. An
embedding of a graph $G$ on a surface is said to be \emphdef{polyhedral}
if $G$ is 3-connected and the smallest length of a non-contractible
noose is at least 3 or if $G$ is a clique and it has at most 3
vertices. In particular, a polyhedral embedding is cellular. If $G$ is
a graph embedded on $\Sigma$, the surface $\Sigma'$ obtained by
\emphdef{cutting} $\Sigma$ along $G$ is the disjoint union of the faces of
$G$, it is a (a priori disconnected) surface with boundary. When we
cut a surface along a set of nooses, viewed as a graph, the resulting
connected components will be called \emphdef{regions}. A
\emphdef{combinatorial map} of an embedded graph is the combinatorial
description of its embedding, namely the cyclic ordering of the edges
around each vertex.

Given an embedded graph $G$, the \emphdef{medial graph} $M_G$ is the
embedded graph obtained by placing a vertex $v_e$ for every edge $e$
of $G$, and connecting the vertices $v_e$ and $v_e'$ with an edge
whenever $e$ and $e'$ are adjacent on a face of $G$. The
\emphdef{barycentric subdivision} of an embedded graph $G$ is the
embedded graph obtained by adding a vertex on each edge and on each
face and an edge between every such face vertex and its adjacent
(original) vertices and edge vertices.

For $\Sigma$ a surface and $G$ a graph embedded on $\Sigma$, a
\emphdef{cut graph} of $(\Sigma,G)$ is a subgraph $H$ of $G$ whose unique
face is a disk. The length of the cut graph is the sum of the lengths
of the edges of $H$. Throughout the whole paper, OPT will denote the length of the shortest
cut graph of $(\Sigma,G)$. 



We refer the reader to~\cite{bdt-ptass-13,rst-dpgs-14} for
definitions pertaining to \emphdef{tree decomposition} and
\emphdef{branch decomposition}.  A \emphdef{carving decomposition} of
a graph $G$ is the analogue of a branch decomposition with vertices
and edges inverted, with the \emphdef{carving-width} defined
analogously. A \emphdef{bond carving decomposition} is a special kind
of carving decomposition where the middle sets always separate the
graph in two connected components. Since these concepts only appear
sporadically in this paper, we refer to~\cite{rst-dpgs-14} for a
precise definition.

\paragraph{Mortar graph and bricks.} The framework of mortar graphs
and bricks has been developed by Borradaile, Klein and
Mathieu~\cite{bkm-nasst-09} to efficiently compute spanners for subset
connectivity problems in planar graphs. 
We recall here
the main definitions around mortar graphs and bricks and refer to the
articles~\cite{bdt-ptass-13,bkm-nasst-09} for more background on these objects.

Let $G$ be a graph embedded on $\Sigma$ of genus $g$.
A path $P$ in a graph $G$ is $\varepsilon$-short in $G$ if for every pair
of vertices $x$ and $y$ on $P$, the distance from $x$ to $y$ along $P$ is
at most $(1+\varepsilon)$ times the distance from $x$ to $y$ in $G$:
$\textrm{dist}_P(x,y) \leq (1+\varepsilon) \textrm{dist}_G(x,y)$. For
$\varepsilon >0$, let $\kappa(g,\varepsilon)$ and $\alpha(g,\varepsilon)$
be functions to be defined later. A \emphdef{mortar graph}
$MG(G,\varepsilon)$ is a subgraph of $G$ such that $\ell(MG) \leq \alpha OPT$,
and the faces of $MG$ partition $G$ into \emphdef{bricks} $B$ that
satisfy the following properties:

\begin{enumerate}\cramped
\item $B$ is planar.
\item The boundary of $B$ is the union of four paths in clockwise order $N,
  E, S$, $W$.
\item $N$ is $0$-short in $B$, and every proper subpath of $S$ is
  $\varepsilon$-short in $B$.
\item There exists a number $k \leq \kappa$ and vertices $s_0 \ldots s_k$ ordered from left to right along $S$ such that, for any vertex $x$ of $S[s_i,s_{i+1})$, $\dist_S(x,s_i) \leq \varepsilon \dist_B(x,N)$.

\end{enumerate}

The mortar graph is computed using a slight variant of the procedure
in~\cite[Theorem~4]{bdt-ptass-13}, the idea is the following:
\begin{enumerate}\cramped
 \item Cut $\Sigma$ along an approximate cut graph, yielding a disk $D$
   with boundary $\partial D$.
 \item Find shortest paths between certain vertices of $\partial D$. This
   defines the $N$ and $S$ boundaries of the bricks.
 \item Find shortest paths between vertices of the previous paths. These
   paths are called the columns.
 \item Take every $\kappa$th path found in the last step. These paths are
   called the \emphdef{supercolumns} and form the $E$ and $W$ boundaries of
   the bricks. The constant $\kappa$ is called the \emphdef{spacing} of the supercolumns.
 \end{enumerate}

This leads to the following theorem to
compute the mortar graph in time $O(g^2 n \log n)$.

\begin{theorem}\label{T:mortar}
Let $\varepsilon >0$ and $G$ be a graph embedded on $\Sigma$ of genus $g$. There exists $\alpha=O(\log^2g \varepsilon^{-1})$ such that there is a mortar
graph $MG(G,\varepsilon)$ of $G$ such that $\ell(MG) \leq \alpha OPT$ and
the supercolumns of $MG$ have length $\leq \varepsilon OPT$ with spacing
$\kappa=O(\log^2g \varepsilon^{-3})$. This
mortar graph can be found in $O(g^2n \log n)$ time.
\end{theorem}

The proof of Theorem~\ref{T:mortar} relies on the following planar construction of the
mortar graph obtained by
Borradaile, Klein and Mathieu~\cite{bkm-nasst-09} (we cite the version
of Borradaile, Demaine and Tazari~\cite[Theorem~2]{bdt-ptass-13}.)

\begin{theorem}\label{T:BKM}
Let $\varepsilon >0$ and $G$ be a planar graph with outer face $H$,
such that $\ell(H) \leq \alpha_0 OPT$, for some $\alpha_0$. For
$\alpha=(2\alpha_0+1)\varepsilon^{-1}$, there is a mortar graph $MH(G,
\varepsilon)$ containing $H$ whose length is at most $\alpha OPT$ and
whose supercolumns have length $\varepsilon OPT$ with spacing
$\kappa=\alpha_0\varepsilon^{-2}(1+\varepsilon^{-1})$. The mortar graph can be found
in $O(n \log n)$ time.
\end{theorem}

\begin{proof}[Proof of Theorem~\ref{T:mortar}] The proof of this
  theorem follows closely the one in~\cite{bdt-ptass-13}. The main
  difference is that the value of OPT is different, therefore we use a
  different cutting strategy from the start.

 We first compute an $O(\log^2 g)$ approximation of the optimal
 cut graph using the algorithm in of Erickson and
 Har-Peled~\cite{eh-ocsd-04} and cut along it, to obtain a planar graph
 $H$. We can now apply Theorem~\ref{T:BKM} to $H$ to obtain a mortar
 graph of $H$, and it is easy to verify that it is a mortar graph of
 $G$ as well. The theorem follows from the bounds of
 Theorem~\ref{T:BKM} and the $O(\log^2 g)$ approximation, the
 bottleneck of the complexity being the computation of the
 approximation of the cut graph.
\end{proof}

\section{Structure Theorem}\label{S:struc}
\label{sect:struct_thm}

\begin{figure}
  \centering
  \includegraphics[scale=0.4]{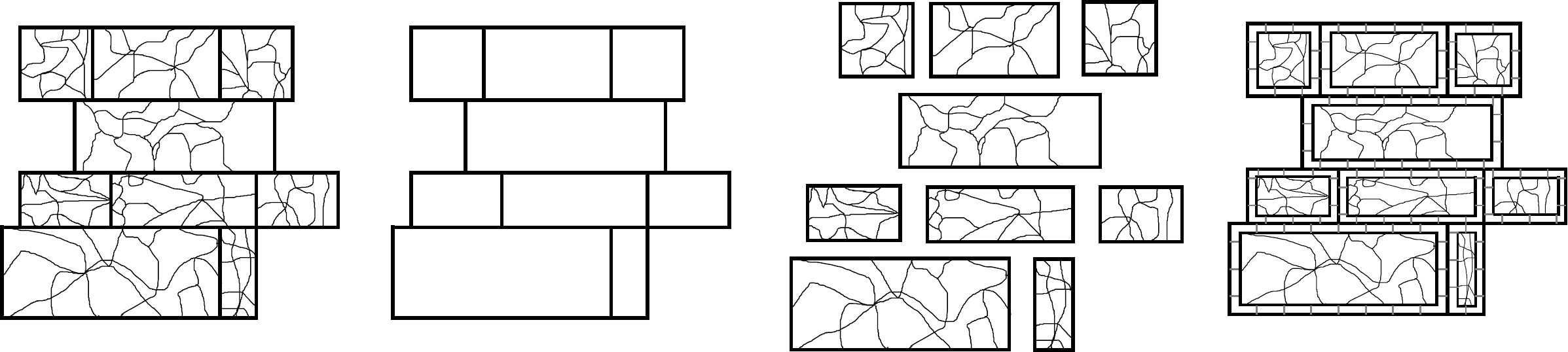}
  \caption{The different stages of the brick decomposition of a graph $G$, the mortar graph, the set
  of bricks and the graph $B^+(MG,\theta)$.}

  \label{fig:bricks}
\end{figure}

In this section, we prove the structure theorem, which shows that
there exists an $\varepsilon$-approximation to the optimal cut graph
which only crosses the mortar graph at a small subset of vertices
called \textit{portals}.

In order to state this theorem, following the literature, we define a
\textit{brick-copy} operation $B^+$ as follows. For each brick $B$, a
subset of $\theta$ vertices is chosen as \textit{portals} such that
the distance along $\partial B$ between any vertex and the closest
portal is at most $\ell(\partial B)/\theta$. For every brick $B$,
embed $B$ in the corresponding face of $MG$ and connect every portal
of $B$ to the corresponding vertex of $MG$ with a zero-length
\textit{portal edge}; this defines $B^+(MG,\theta)$. The edges
originating from $MG$ are called the \textit{mortar edges}.

We note that by construction, $B^+(MG,\theta)$ embeds on the plane in such a way
that every brick of $B^+(MG,\theta)$ is included in the corresponding brick
of $MG$. Furthermore, every vertex of $G$ corresponds to a vertex of
$B^+(MG,\theta)$ by mapping the insides of bricks to the insides of
bricks in $B^+(MG,\theta)$, and the mortar graph to itself,
cf. Figure~\ref{fig:bricks}. We denote this map by $i:V(G) \rightarrow V(B^+(MG,\theta))$.

Moreover, we contract the $E$ and $W$ boundaries of each brick of
$B^+(MG,\theta)$ and their copies in the mortar graph.
Since the sum of the length of the $E$ and $W$ boundaries is at most $\varepsilon OPT$,
any solution of length $\ell$ in $B^+(MG,\theta)$ going through a vertex resulting from
a contraction can be transformed into a solution of length at most $\ell + 2 \varepsilon OPT$ in
$B^+(MG,\theta)$ where no edge is contracted.
The structure theorem is then the following:

\begin{theorem}\label{T:structheorem}
Let $G$ be a graph embedded on $\Sigma$ of genus $g$, and $\varepsilon
>0$. Let $MG(G,\varepsilon)$ be a corresponding mortar graph of length
at most $\alpha OPT$ and supercolumns of length at most $\varepsilon
OPT$ with spacing $\kappa$. There exists a constant $\theta(\alpha,
\varepsilon,\kappa)$ depending polynomially on $\alpha, 1/\varepsilon$
and $\kappa$ such that:
\[\OPT(B^+(MG,\theta)) \leq (1 + c\varepsilon) \OPT.\]

\end{theorem}

The proof of this theorem essentially consists in plugging in the
structure theorem of~\cite{bkm-nasst-09} and verifying that it fits. Let us
first recall the structural theorem of bricks~\cite{bkm-nasst-09}. For a
graph $H$ and a path $P \subseteq H$, a \textit{joining vertex} of $H$ with
$P$ is a vertex in $P$ that lies on an edge of $H \setminus P$.


\begin{theorem}\cite[Theorem~10.7]{bkm-nasst-09}\label{T:strucbricks}
Let $B$ be a brick, and $F$ be a set of edges of $B$. There is a forest
$\widetilde{F}$ in $B$ with the following properties:
\begin{enumerate}\cramped
\item If two vertices of $N \cup S$ are connected by $F$, they are also
  connected by $\widetilde{F}$,
\item The number of joining vertices of $\widetilde{F}$ with both $N$ and
  $S$ is bounded by $\gamma(\kappa, \varepsilon)$,
\item $\ell(\widetilde(F))\leq \ell(F)(1+c \varepsilon)$.

In the above, $\gamma(\kappa,\varepsilon)=o(\varepsilon^{-2.5}\kappa)$ and $c$ is a
fixed constant.
\end{enumerate}
\end{theorem}

From this we can deduce the following proposition.

\begin{proposition}\label{P:strucBKM}
Let $C$ be a subgraph of $G$ of length $OPT$. There exists a constant
$\theta(\alpha, \varepsilon, \kappa)$ depending polynomially on
$\alpha,1/\varepsilon$ and $\kappa$ and a subgraph
$\widehat{C}$ of $B^+(MG,\theta)$ with the following properties:
\begin{itemize}\cramped
\item $\ell(\widehat{C}) \leq (1+\tilde{c}\varepsilon) \ell(C)
  =(1+\tilde c\varepsilon) OPT$, where $\tilde c$ is a fixed constant.
\item If we denote by $D$ the closed disk on which $MG$ has been
  constructed, for any two vertices $s,t \in \partial D$ that are
  connected by $C$ in $D$, $i(s)$ and $i(t)$ are connected by $\widehat{C}$ in $D$ as well.
\end{itemize}
\end{proposition}

\begin{proof}
We recall that the interior of every brick $B$ of $G$ can naturally be embedded in
the corresponding brick of $B^+(MG,\theta)$, therefore, for every
brick $B$, we can identify $C \cap B$ with a subgraph of
$B^+(MG,\theta)$. We define $\widehat{C}$ as follows: 
\begin{itemize}
\item For every edge of $C$ which is an edge of the mortar graph $MG$,
  we add the corresponding mortar edge in $B^+(MG,\theta)$ to $\widehat{C}$.
\item In the inside of every brick $B$, if $F$ is the
  forest induced by $C \cap B$, we define $\widehat{C} \cap B$ to be
  the forest defined by $\widetilde{F}$ (as defined in Theorem
  \ref{T:strucbricks})
\item Finally, for every brick $B$, let $Q(B)$ denote the set of
  joining vertices of $\widetilde{C}$ with $N \cup S$. For every
  vertex of $Q(B)$ in a brick $B$, we add to $\widetilde{C}$ the path to
  its closest portal, the copy of this path in the mortar graph, as well as the corresponding portal edge.
\end{itemize}

  We start by proving the first property. The
  length of $\widehat{C}$ restricted to $B$ is at most
  $(1+c\varepsilon)$ times the length of $C$ restricted to $B$.  The
  length incurred by the portal operation is
  \[\sum\limits_{B} \sum\limits_{q \in Q(B)} \le 2\sum\limits_{B}  |Q(B)| \frac{\ell(\partial B)}{\theta} \le 2\sum\limits_{B} \gamma(\varepsilon) \frac{\ell(\partial B)}{\theta}.\]
  By defining $\theta = \gamma(\varepsilon) \alpha\varepsilon^{-1}$
  portals, we obtain that the length of $\widehat{C}$ is at most
  $(1+c'\varepsilon) OPT$ for some universal constant $c'$, since the
  length of the mortar graph is at most $\alpha OPT$.

  We now prove the second property.  Consider two vertices $u,v$ on
  $\partial D$ that are connected in $C$ and let $P$ be the path in
  $C$ connecting them.  Recall that any vertex that belongs to
  $\partial D$ belongs to the mortar graph.  We decompose $P$ into
  subpaths crossing at most one brick and whose extremities lie on the
  mortar graph.  For any such subpath $s = \{s, \ldots, t\}$ in a brick $B$, we
  show that there exists a path in $\widehat{C}$ connecting $i(s)$
  to $i(t)$.  We consider the set of vertices $s_B := Q(B) \cap s$. By
  definition of $\widehat{C}$, $i(s)$ and $i(t)$ are connected
  to vertices of $s_B$.  Property~1 of Theorem \ref{T:strucbricks}
  implies that for any couple of vertices of $Q(B)$, if they are
  connected in $C$ they are also connected in
  $\widehat{C}$. Therefore, we conclude that $i(s)$ and $i(t)$ are
  connected and therefore $u$ and $v$ are also connected.

\end{proof}

We now have all the tools to prove the structure theorem.

\begin{proof}[Proof of Theorem~\ref{T:structheorem}]
Let $C$ be an optimal cut graph of $(\Sigma,G)$. We apply
Proposition~\ref{P:strucBKM} to $C$, it yields a subgraph
$\widehat{C}$ of $B^+(MG,\theta)$ of length $\ell(\widehat{C})\leq
(1+\varepsilon) \ell(C)$. We claim that this graph $\widehat{C}$
contains a cut graph of $\Sigma$.

Suppose on the contrary that there exists a non-contractible cycle
$\gamma$ in $(\Sigma,B^+(MG,\theta))$ which does not cross $\widehat{C}$. This cycle
$\gamma$ corresponds to a cycle $\gamma'$ in $(\Sigma,G)$ by contracting
portal edges, and since $C$ is a cut graph, there exists a maximal
subpath $P$ of $C$ restricted to $D$ and a maximum subpath $P'$ of
$\gamma'$ in $D$ such that $P'$ crosses $P$ an odd number of times,
otherwise, by flipping bigons we could find a cycle homotopic to
$\gamma'$ not crossing $C$. Denote by $(s,t)$ and $(s',t')$ the
intersections of $P$ and $P'$ with $\partial D$. Then, without loss of
generality, $s$, $s'$, $t$ and $t'$ appear in this order on $\partial
D$. Furthermore, the vertices $i(s)$ and $i(t)$ in $B^+(MG,\theta)$
are connected by $\widehat{C}$ by Proposition~\ref{P:strucBKM}, since
$s$ and $t$ are connected by $C$. Therefore, $\gamma$ crosses
$\widehat{C}$, and we reach a contradiction.
\end{proof}

\section{Algorithm}\label{S:algorithm}

We now explain how to apply the spanner framework of Borradaile et al.
in~\cite{bdt-ptass-13} to compute an
approximation of the optimal cut graph. We start by computing the
optimal Steiner tree, for each subset of the portals in every brick by
using the algorithm or Erickson, Monma, and Veinott ~\cite{emv-ssmmc-87}, and then 
take the union of all these trees over all bricks, plus the edges of the 
mortar graph. As this algorithm runs in time $O(nk^3)$, this step takes time 
$O_{g,\varepsilon}(n)$.
This defines the graph $G_{span}$, which by construction has length $\leq
f(g,\varepsilon) OPT$, where $f(g,\varepsilon)=O(2^{\theta})=2^{O(\log^2g) poly(1/\varepsilon)}$, and contains a $(1+\varepsilon)$
approximation of the optimal cut graph by the structure theorem.

We will use the following theorem of Demaine et al.~\cite[Theorem
  1.1]{dhm-aacd-10} (the complexity of this algorithm can be improved to $O_g(n)$~\cite{dhk-cthmfg-11}).

\begin{theorem}\label{T:tw}
  For a fixed genus $g$, any $k \geq 2$ and any graph $G$ of genus at
  most $g$, the edges of $G$ can be partitioned into $k$ sets such that
  contracting any one of the sets results in a graph of tree-width at most
  $O(g^2k)$. Furthermore, the partition can be found in time
  $O(g^{5/2}n^{3/2} \log n)$ time.

\end{theorem}

The four steps of the framework are now the following.

\begin{enumerate}\cramped

\item Compute the spanner $G_{span}$.
\item Apply Theorem~\ref{T:tw} with $k=f(g,\varepsilon) / \varepsilon$, and
  contract the edges in the set of the partition with the least weight. The
  resulting graph $G_{tw}$ has tree-width at most $O(g^2 \varepsilon^{-1}
  f(g,\varepsilon))$.
\item Use the bounded tree-width to compute a cut graph of
  $(\Sigma,G_{tw})$. An algorithm to do this is described in Section~\ref{S:tw}.
\item Incorporate the contracted edges back. By definition, they have
  length at most $f(g,\varepsilon)OPT / k = \varepsilon OPT$. Therefore,
  the final graph we obtain has the desired length. If the
  resulting graph has more than one face, remove edges until we obtain
  a cut-graph.
\end{enumerate}

We now analyze the complexity of this algorithm. The spanner is computed in
time $O_{g,\varepsilon}(n \log n)$. Using~\cite{dhk-cthmfg-11}, the second step takes time $O_{g,\varepsilon}(n)$.
Dynamic programming takes time $O_{g,\varepsilon}(n^3)$ (see thereafter), and the final lifting step
takes linear time. Assuming the dynamic programming step described in
the next section, this proves Theorem~\ref{T:main}.

\section{Computing cut graphs for bounded tree-width}\label{S:tw}

There remains to prove that computing the optimal cut graph of
$(\Sigma,G)$ is fixed parameter tractable with respect to both the
tree-width of $G$ and the genus of $\Sigma$ as a parameter. Out of convenience,
we work with the branch-width instead, which gives the result since they are
within a constant factor~\cite[Theorem~5.1]{rs-gm10otd-91}. As cut
graphs are a topological object, we
will rely on surface-cut decompositions~\cite{rst-dpgs-14}, which are a topological
strengthening of branch decompositions. Note that, for reasons which
will be clear later, our definition is slightly
different from the one of Ru\'e, Sau and Thilikos as it does not rely
on polyhedral decompositions.


Given a graph $G$ embedded in a surface $\Sigma$ of genus $G$, a
\textit{surface-cut decomposition} of $G$ is a branch decomposition
$(T, \mu)$ of $G$ such that for each edge $e \in E(T)$, the vertices
in $\textrm{mid}(e)$ are contained in a set $\mathcal{N}$ of nooses in
$\Sigma$ such that:

\begin{itemize}\cramped
\item $|\mathcal{N}|=O(g)$
\item The nooses in $\mathcal{N}$ pairwise intersect only at subsets of $\textrm{mid}(e)$
\item $\theta(\mathcal{N})=O(g)$
\item $\Sigma \setminus \bigcup \mathcal{N}$ contains exactly two connected
  components, of which closures contain respectively $G_1$ and $G_2$.
\end{itemize}
where $\theta$ is defined as follows: for a point $p$ in $\Sigma$, if we denote by $\mathcal{N}(p)$
the number of nooses in $\mathcal{N}$ containing $p$, and let
$P(\mathcal{N}) = \{p \in \Sigma \mid \mathcal{N}(p) \geq 2 \}$, we define
\[\theta(\mathcal{N})=\sum_{p\in P(\mathcal{N})}\mathcal{N}(p) -1.\]
Ru\'e et al. showed how to compute such a surface-cut
decomposition when the input graph $G$ is embedded
\textbf{polyhedrally} on the surface $\Sigma$:

\begin{theorem}[{\cite[Theorem~7.2]{rst-dpgs-14}}]\label{T:rst}
Given a graph $G$ on $n$ vertices polyhedrally embedded on a surface
of genus $g$ and with $bw(G) \leq k$, one can
compute a surface-cut decomposition of $G$ of width $O(g+k)$ in time $2^{O(k)}n^3$.
\end{theorem}


When the input graph is not polyhedral, Ru\'e et al.
propose a more intricate version of surface-cut decompositions relying
on \textit{polyhedral decompositions}, but it is unclear how to
incorporate these in a dynamic program to compute optimal cut graphs.

Instead, we present two ways to circumvent polyhedral decompositions
and use these surface-cut decompositions directly. The first one
consists of observing that the difficulties involved with computing
surface-cut decompositions of non-polyhedral embeddings can be
circumvented by using a theorem of Inkmann~\cite{i-tpdgs-07}. Since
Inkmann's theorem has, up to our knowledge, not been published outside
of his thesis, and the proof is quite intricate, for the sake of
clarity we also provide a different approach, based on modifying the
input graph to make it polyhedral.

In both cases, we obtain a branch decomposition with a strong
topological structure, which we can then use as a basis for a dynamic
program to compute the optimal cut graph.

\subsection{A simpler version of surface-cut decompositions}

The algorithm~\cite[Algorithm~2]{rst-dpgs-14} behind the proof of Theorem~\ref{T:rst} relies on the
following steps. Starting with a polyhedral embedding of $G$ on a surface, 

\begin{enumerate}\cramped
\item Compute a branch decomposition \textit{branch(G)} of $G$.
\item Transform \textit{branch(G)} into a carving
  decomposition \textit{carv(G)} of $M_G$.
\item Transform \textit{carv(G)} into a \textit{bond}
  carving decomposition \textit{bond(G)} of $M_G$.
\item Transform \textit{bond(G)} into a branch decomposition of $G$.
\end{enumerate}

The second step is the only one where where the polyhedrality of the
embedding is used, as it relies on the following lemma:

\begin{lemma}[{\cite[Lemma~5.1]{rst-dpgs-14}}]
Let $G$ be a polyhedral embedding on a surface $\Sigma$ of genus $g$, and $M_G$ be the
embedding of the medial graph. Then $bw(G)\leq cw(M_G)/2 \leq
6bw(G)+4g+O(1)$, and the corresponding carving decomposition of $M_G$
can be computed from the branch decomposition of $G$ in linear time.
\end{lemma}

We observe that the following theorem of Inkmann shows that the branch-width of a surface-embedded graph and
the carving-width of its medial graph are tightly related, even for
non-polyhedral embeddings.

\begin{theorem}[{\cite[Theorem~3.6.1]{i-tpdgs-07}}]\label{T:inkmann}
For every surface $\Sigma$ there is a non-negative constant $c(\Sigma)$ such that if
$G$ is embedded on $\Sigma$ with $|E(G)| \geq 2$ and $M_G$ is its medial graph, we
have
$2bw(G) \leq cw(M) \leq 4bw(G)+c(\Sigma).$
\end{theorem}


Digging into the proof reveals that $c(\Sigma)=O(g^2)$ for $\Sigma$ of genus
$g$.  The idea is therefore that replacing Lemma~5.1 of Ru\'e et
al. by Theorem~\ref{T:inkmann} allows us to lift the requirement of
polyhedral embedding in their construction. One downside is that this
theorem does not seem to be constructive, and therefore we need an
alternative way to compute the carving decomposition in step 2. This
can be achieved in fixed parameter tractable time with respect to the
carving-width (and linear in $n$) using the algorithm of Thilikos,
Serna and Bodlaender~\cite{tsb-cltasc-00}. In conclusion, we obtain
the following corollary (note that the bottleneck in the complexity is
the same as in the one of Ru\'e et al., which is the transformation
between a carving and a bond carving decomposition).

\begin{corollary}\label{C:sc-decompositions}
Given a graph $G$ on $n$ vertices embedded in a surface of genus $g$
with $\textrm{bw}(G) \leq k$, there exists an algorithm running in
time $O_k(n^3)$ computing a surface-cut decomposition $(T,\mu)$ of $G$
of width at most $O(k+g^2)$.
\end{corollary}

\subsection{Making a graph polyhedral}
\label{sect:making_polyhedral}
In this section, we show how to go from an embedded graph to a
polyhedral embedding, without increasing the tree-width too much. The
construction will be split in the following three lemmas.

\begin{lemma}
\label{lem:triangul}
Let $G$ be a graph of tree-width at most $k\geq 2$, embedded on a surface
of genus $g$. Then there exists a triangulation of $G$ of tree-width
at most $k$. Moreover, given a tree-decomposition of width $k$, one can
compute a triangulation of $G$ of tree-width at most $k$ in polynomial time.
\end{lemma}

\begin{proof}
Let $\widetilde{G}$ be the chordal graph containing $G$ having the
smallest clique number, i.e., $k+1$. We just observe that
$\widetilde{G}$ contains a triangulation of $G$, which will therefore
have the same treewidth as $G$.

We now show how to compute a triangulation of tree-width $k$ given a tree-decomposition
$\mathcal{T}$ of width $k$.
Consider the graph $\widetilde{G}$ which consists of $G$ plus all the edges connecting
the vertices $u,v$ that are present in the same bag of $\mathcal{T}$.
By definition of the tree-decomposition, $\widetilde{G}$ is chordal and has the same 
tree-width than $G$.
Then, for any cycle of $G$ of length at least 4, there exists a bag which induces 
at least one chord, namely such that adding all the edges between the vertices of the bag
creates a chord in the cycle. Therefore, we can proceed greedily and for each face $f$ of length
at least 4 of $G$, find a bag that contains two non-consecutive vertices of $f$, add
this edge to $G$, embed it in the face and proceed recursively.
\end{proof}

\begin{lemma}
\label{lem:baryc}
Let $G$ be a triangulated graph of tree-width at most $k$, embedded on
a surface of genus $g$. Then its barycentric subdivision $B(G)$ has tree-width at
most $f(k,g)$ for some function $f$.
\end{lemma}

\begin{proof}
The barycentric subdivision consists of the \textit{original vertices}
of $G$, the \textit{edge vertices} and the \textit{face vertices}.
Let $G_e$ be the barycentric subdivision of $G$ restricted to the face
and edge vertices, namely $G_e$ consists in the dual of $G$ where each
edge is subdivided.  Let $T=(X_1 \ldots X_n)$ be a tree decomposition
of $G_e$.  By~\cite{m-twhsd-12}, the size of the bags of $T$ is
bounded by some function $h(k,g)$. We consider the tree-decomposition
$T'=(X_1' \ldots X_n')$ of $B(G)$ obtained by adding an original
vertex to every bag $X_i$ containing at least one of its neighbors in
its barycentric subdivision, namely an adjacent face or edge vertex.
Since $G$ is triangulated, the size of the bags is multiplied by at
most a constant. Let us prove that $T'$ is a tree decomposition:

\begin{itemize}
\item It contains all the vertices of $B(G)$.
\item For every edge, there is a bag containing both endpoints.
\item Let $X_i'$ and $X_j'$ be two bags containing a vertex $v$. If
  $v$ is a face or an edge vertex, every bag on the path between
  $X'_i$ and $X'_j$ contains it. If it is an original vertex, then
  $X'_i$ and $X'_2$ both contain a neighbor of $v$, which we denote
  by $v_1$ and $v_2$. Now, if there is a bag $X'_k$ on the path
  between $X'_i$ and $X'_j$ which does not contain $v$, it does not
  contain any neighbor of $v$ either, but it separates $v_1$
  from $v_2$ in $B(G)$. We have reached a contradiction.
\end{itemize}

\end{proof}

\begin{lemma}
\label{lem:boundedtw}
Let $G$ be a triangulated graph of tree-width at most $k$, embedded on a surface
of genus $g$. Let $M_G$ denote the medial graph of $G$, and $G'$ the
superposition of $G$ and $M_G$. Then the tree-width of $G'$ is
bounded by some function of $k$ and $g$.
\end{lemma}

\begin{proof}
The idea of the proof is the same as for the previous lemma. We start
with a tree decomposition $T=(X_1 \ldots X_n)$ of $M_G$. Since $M_G$ is the dual of the
radial graph of $G$, which is contained in the barycentric subdivision
of $G$, by the previous lemma and~\cite{m-twhsd-12}, the treewidth
of $M_G$ is bounded by some function of $k$ and $g$. Now, define
$T'=(X'_1 \ldots X'_n)$ by adding every original vertex in $G'$ to the
bags of $T$ containing any of its neighbors. The size of the bags at
most triple in size, and let us prove that $T'$ is a tree
decomposition:

\begin{itemize}
\item It contains all the vertices of $G'$.
\item For every edge, there is a bag containing both endpoints.
\item Let $X_i'$ and $X_j'$ be two bags containing a vertex $v$. If
  $v$ is a vertex of $M_G$, every bag on the path between
  $X'_i$ and $X'_j$ contains it. If it is an original vertex, then
  $X'_i$ and $X'_2$ both contain a neighbor of $v$, which we denote
  by $v_1$ and $v_2$. Now, if there is a bag $X'_k$ on the path
  between $X'_i$ and $X'_j$ which does not contain $v$, it does not
  contain any neighbor of $v$ either, but it separates $v_1$
  from $v_2$ in $G'$. We have reached a contradiction.
\end{itemize}

\end{proof}

Now, we observe that superposing medial graphs two times
increases the length of non-contractible nooses of a graph. Furthermore, if the new edges are
weighted heavily enough (e.g., with a weight larger than OPT, which we know
how to approximate), they will not change the value of the optimal
cut graph\footnote{When an edge is cut in two halves, the weight is spread
  in half on each sub-edge.}. Therefore this allows us to assume that the embedded graph
of which we want to compute an optimal cut graph has only
non-contractible nooses of length at least three. By subdividing it to remove loops and multiple edges and triangulating it, we can also assume that it is
3-connected (since the link of every vertex of a triangulated simple graph is
2-connected), and therefore that it is polyhedral.

For a polyhedral embedding, our definition of surface-cut
decompositions and the one of Ru\'e et al.~\cite{rst-dpgs-14} coincide, and therefore
we can use their algorithm to compute it.

\subsection{Dynamic programming on surface-cut decompositions}





We now show how to compute an optimal cut graph of an embedded graph
of bounded branch-width, using surface-cut decompositions. We first
recall the following lemma of Erickson and Har-Peled which follows
from Euler's formula and allows us to
bound the complexity of the optimal cut graph. For a graph $H$
embedded on a region $R$, we define its \textit{reduced graph} to be
the embedded graph obtained by repeatedly removing from $G$ the degree
1 vertices which are not on a boundary and their adjacent edges, and
contracting each maximal path through degree 2 vertices to a
single edge (weighted as the length of the path). 

\begin{lemma}[{\cite[Lemma~4.2]{eh-ocsd-04}}]\label{L:reduced}
Let $\Sigma$ be a surface of genus $g$. Then any reduced cut graph on $\Sigma$
has less than $4g$ vertices and $6g$ edges.
\end{lemma}

The idea is then to compute in a dynamic programming fashion, for
every region $R$ of the surface-cut decomposition, every possible
combinatorial map $M$ corresponding to the restriction of a reduced
cut graph of $\Sigma$ to $R$, and every possible position $P$ of the vertices of the boundary of
$M$ on the boundary of $R$, the shortest reduced graph embedded on $R$ with
map $M$ and position $P$. The bounds on the size of the boundaries of
the region (coming from the definition of surface-cut decompositions),
as well as Lemma~\ref{L:reduced} allow us to bound the size of the
dynamic tables.

\begin{theorem}\label{T:boundedtw}
  If a graph $G$ of complexity $n$ embedded on a genus $g$ surface has
  branch-width at most $k$, an optimal cut graph of $G$ can be computed
  in time $O_{g,k}(n^3)$.
\end{theorem}

\begin{proof}[Proof of Theorem~\ref{T:boundedtw}]
  We start by computing a surface-cut decomposition $(T,\mu)$ of
  $(\Sigma,G)$ using either of the algorithms presented in
  Section~\ref{S:tw}, and the width of $(T,\mu)$ is $O(g^2+k)$.

  Then, our algorithm relies on dynamic programming. For every edge
  $e$ in the tree of the surface-cut decomposition, there is a set of
  nooses $\mathcal{N}_e$, such that cutting $\Sigma$ along $\mathcal{N}_e$
  yields two connected regions $R_1$ and $R_2$ of $\Sigma$. The set of
  nooses $\mathcal{N}_e$ contains exactly $mid(e)$ vertices, but every
  vertex appearing multiple times in $\mathcal{N}_e$ gets copied as
  many times when considered on the boundary of $R_1$ or
  $R_2$. However, since $\theta(\mathcal{N}_e)=O(g)$, this only
  happens $O(g)$ times at most, and therefore the boundary of $R_1$ (or $R_2$) contains
  $O(g+mid(e))=O(g^2+k)$ vertices. 

For any region $R$ of the surface-cut decomposition, a reduced
cut graph $C$ of $(\Sigma,G)$ induces a combinatorial map $M$ on $R$. We
denote by $\mathcal{M}(R)$ the set of all these maps, for all the
reduced cut graphs of $(\Sigma,G)$. For a map $M$ in $\mathcal{M}(R)$, the
vertices of $M$ on the boundary of $R$ are called its \textit{boundary
  vertices}. Any embedded graph on $R$ with map $M$ induces a matching
between the boundary vertices of $M$ and the vertices on the boundary
of $R$, and the set of all these possible matchings is called the set
$\mathcal{P}(M,R)$ of \textit{boundary positions} of $M$ in $R$.


The \textit{reduced combinatorial map} of an embedded graph is the
combinatorial map of its reduced graph. For every region $R$ induced
by the surface-cut decomposition, every combinatorial map $M$ in
$\mathcal{M}(R)$ and every boundary position $P\in \mathcal{P}(M,R)$,
the dynamic programming tables store a number $T[R,M,P]$ which is the
length of the shortest subgraph of $G$ in $R$ with reduced
combinatorial map $M$ and with boundary positions $P$. Given a region
$R$ and its two subregions $R_1$ and $R_2$, the data of the dynamic
programming tables of $R_1$ and $R_2$ allow to compute the table of
$R$ by the following formula:

\[T[R,M,P]= \min_{M_1,M_2\in \mathcal{S}(M,R_1,R_2)} \min_{P_1, P_2\in \mathcal{T}(M_1,M_2,P)} T[R_1,M_1,P_1]+T[R_2,M_2,P_2],\]
where $\mathcal{S}(M,R_1,R_2) \subseteq \mathcal{M}(R_1) \times
\mathcal{M}(R_2)$ is the set of combinatorial maps in $R_1$ and $R_2$,
which glued together give the map $M$ on $R$, and
$\mathcal{T}(M_1,M_2,P)\subseteq \mathcal{P}(M_1,R_1) \times
\mathcal{P}(M_2,R_2)$ is the set of boundary positions of $M_1$ in
$R_1$ and $M_2$ in $R_2$ such that vertices glued together on the boundary of $R_1$
and $R_2$ are mapped to the
same vertex and vertices on the boundary of $R$ are mapped according to
$P$. As usual, the minimum is taken to be infinite if the sets are empty.

We now bound the size of the tables. For a region $R$, the set
$\mathcal{M}(R)$ consists of combinatorial maps having at most $4g$
vertices of degree at least 3 (by Lemma~\ref{L:reduced}), no vertices
of degree 2 (since they have been contracted during the reduction),
and vertices of degree 1 only on the boundary of $R$. Since the number
of vertices on the boundary of $R$ is $O(g^2+k)$, the size of
$\mathcal{M}(R)$ is bounded by a function of $g$ and $k$. Similarly,
the size of $\mathcal{P}(M,R)$ is bounded by a function depending on
the number of boundary vertices of $M$ and $R$, and thus by a function
of only $g$ and $k$.

Finally, the length of the optimal reduced cut graph is equal to the minimum
of $T[\Sigma,M,\emptyset]$, where $M$ ranges over all the combinatorial
maps of reduced cut graphs on $\Sigma$. The number of such combinatorial
maps is bounded a function
of $g$ by Lemma~\ref{L:reduced}. The complexity of the dynamic program
is therefore $O_{g,k}(n^3)$, where the cubic dependency comes from the
computation of the surface-cut decomposition. Since, when doing a
reduction, removing the degree 1 vertices not
on the boundary only reduces the length, the length of the optimal
reduced cut graph is the same as the length of the optimal cut graph. This concludes the proof.

\end{proof}

\paragraph{Open problems.}One of the main challenges is whether the problem of computing the shortest
cut graph can be solved \textit{exactly} in FPT complexity -- the
recent application of brick decompositions to exact solutions for
Steiner problems~\cite{ppsj-nssppb-14} might help in this direction.
In the approximability direction, it is also unknown whether 
there exists a polynomial time constant factor approximation to this problem, 
or even a PTAS.

\paragraph{\textbf{Acknowledgments}}
We are grateful to Sergio Cabello, \'Eric Colin de Verdière, Frederic
Dorn and Dimitrios M. Thilikos for very helpful remarks at various stages of the
elaboration of this article.

\bibliographystyle{siam}
\bibliography{bibexport,arnaud}

\end{document}